\documentclass[11pt]{article}

\usepackage[a4paper,margin=1in]{geometry}
\usepackage{amsmath,amssymb,amsthm,mathtools}
\usepackage{microtype}
\usepackage{xcolor}
\usepackage[colorlinks=true,linkcolor=blue,citecolor=blue,urlcolor=blue]{hyperref}
\usepackage[
  backend=biber,
  style=numeric,
  sorting=none,
  doi=true,
  url=true,
  isbn=false,
  eprint=false
]{biblatex}
\newcommand{\extarticlelink}[1]{%
  \iffieldundef{doi}
    {\iffieldundef{url}
       {#1}
       {\href{\thefield{url}}{#1}}}
    {\href{https://doi.org/\thefield{doi}}{#1}}}

\DeclareFieldFormat{doi}{\mkbibacro{DOI}\addcolon\space\href{https://doi.org/#1}{\nolinkurl{#1}}}
\DeclareFieldFormat{url}{\mkbibacro{URL}\addcolon\space\url{#1}}
\DeclareFieldFormat{bibhyperref}{\extarticlelink{#1}}
\DeclareFieldFormat[article]{title}{\mkbibquote{\extarticlelink{#1}}}
\AtEveryBibitem{\clearfield{month}\clearfield{day}}

\newtheorem{theorem}{Theorem}
\newtheorem{proposition}[theorem]{Proposition}
\newtheorem{remark}[theorem]{Remark}

\newcommand{\R}{\mathbb{R}}
\newcommand{\Ran}{\operatorname{Ran}}
\newcommand{\Var}{\operatorname{Var}}
\newcommand{\cD}{\mathcal{D}}
\newcommand{\ip}[2]{\langle #1,#2\rangle}
\newcommand{\avg}[1]{\langle #1\rangle_0}

\title{Center-of-Mass Bounds and Harmonic Extremality}
\author{Arseny Pantsialei\thanks{e-mail: wselend@gmail.com} \\
\emph{Institute of Physics, Maria Curie-Sk\l{}odowska University, 20-031 Lublin, Poland}}
\date{}

\begin{document}

\maketitle

\begin{abstract}
We study the center-of-mass observable in one-dimensional many-body systems with translation-invariant interactions and extend the harmonic-rigidity mechanism from the 
one-body setting to an interacting many-body problem. We prove a sharp upper bound on the ground-state center-of-mass fluctuation in terms of the active spectral gap 
associated with the center-of-mass probe, and show that this bound does not require any positivity assumption on the ground state. In the positivity class, we characterize 
the equality case completely. Exact saturation occurs if and only if the external one-body traps are harmonic with a common frequency, while the interaction may remain 
arbitrary within the translation-invariant class. We also identify a natural rigidity defect measuring deviation from the harmonic extremal situation and prove quantitative 
near-saturation estimates controlling both the variance deficit and the spectral weight outside the first active shell. In this way, the paper establishes harmonic confinement 
as the unique static extremizer for the rigid interacting center-of-mass mode at fixed active gap.
\end{abstract}

\section{Introduction}

The one-body paper \cite{Pantsialei2026HarmonicRigidity} solved the static isoperimetric
problem behind the Mandelstam--Tamm bound in one dimension:
among confining scalar potentials with a fixed active gap, the harmonic trap is
the unique maximizer of the ground-state position variance. The present paper
asks whether an analogous extremal principle survives in an interacting
many-body system.

At first sight the answer is not obvious. In an $N$-body Hamiltonian, the
position of a single particle is no longer the right observable: interactions
mix coordinates, the first excited state need not be visible to a given probe,
and the relevant gap is therefore observable-dependent. The correct quantity is
instead the center of mass
\[
X=\frac1M\sum_{i=1}^N m_i x_i,
\qquad
M=\sum_{i=1}^N m_i,
\]
and the correct spectral scale is the \emph{active gap}
\[
\Delta_X=\min\{E_n-E_0:\ \ip{\psi_n}{X\psi_0}\neq 0\}.
\]
This is the first excitation energy actually seen by a rigid displacement of the
whole cloud. In particular, $\Delta_X$ need not coincide with the global gap
$E_1-E_0$.

The resulting theorem has a natural conceptual interpretation. Classical
Kohn-type results state that if the external confinement is harmonic and the
interaction is translation invariant, then the center of mass decouples and
oscillates at the trap frequency independently of the interactions
\cite{Kohn1961,Dobson1994,WuZaremba2014,BreyJohnsonHalperin1989}. In trapped-gas
settings, the same rigid-mode theme also appears in the dynamics of
one-dimensional Bose gases and collective oscillations
\cite{Zaremba1999Dynamics,MenottiStringari2002,Moritz2003Collective}. Here we prove a static converse of
extremal type. At the general level we obtain a sharp center-of-mass bound at
fixed active gap, while in the positivity class we show that harmonic
confinement is the unique way to maximize the ground-state center-of-mass
fluctuation. In this sense, harmonic traps are not only dynamically special;
they are also statically extremal.

Further generalized-Kohn and center-of-mass treatments for parabolic
confinement and conserving many-body approximations may be found in
\cite{MaksymChakraborty1990,Peeters1990,Vignale1995,BonitzBalzerVanLeeuwen2007}.

The many-body step is real and not merely formal. The interaction
\[
U\bigl(\{x_i-x_j\}_{i<j}\bigr)
\]
may be arbitrary inside the translation-invariant class, so the theorem is not a
statement about noninteracting particles in disguise. What survives from the
one-body mechanism is the exact center-of-mass algebra
\[
[H,X]=-\frac{i\hbar}{M}P,
\qquad
[X,[H,X]]=\frac{\hbar^2}{M},
\]
which identifies the universal Thomas--Reiche--Kuhn constant for the rigid
displacement mode. The main point of the paper is that this algebra is strong
enough to recover both a sharp variance bound and a rigidity statement in the
interacting problem.

The derivation of the inequality itself is short once the
center-of-mass TRK identity and the active gap are introduced. The main
contribution of the present work is the converse rigidity statement that exact
saturation in the positive-ground-state class implies common-frequency
harmonic confinement, even in the presence of arbitrary
translation-invariant interactions.

This also places the result naturally at the interface of three strands of
literature. First, it belongs to the family of center-of-mass decoupling and
generalized Kohn theorems for harmonically confined interacting systems
\cite{Kohn1961,Dobson1994,WuZaremba2014,BreyJohnsonHalperin1989,Zaremba1999Dynamics}. Second, it uses the model-independent
commutator philosophy behind energy-weighted sum rules and collective-mode
estimates \cite{LippariniStringari1989,Stringari1996,MenottiStringari2002,Moritz2003Collective}. Third, it continues the sharp
active-gap program of \cite{Pantsialei2026HarmonicRigidity}, but in a setting where the
observable, the gap, and the equality structure are genuinely many-body.

For classical and review treatments of the TRK, oscillator-strength,
and energy-weighted sum rules, see
\cite{Kuhn1925,Inokuti1971,BohigasLaneMartorell1979}.

The main conclusions come in two different levels of generality, and it is
important to distinguish them clearly.

Under standard self-adjointness and confinement assumptions, and assuming only
that the interaction depends on relative coordinates, we prove the sharp bound
\[
\Var_0(X)\le \frac{\hbar^2}{2M\Delta_X}.
\]
This part is robust and does \emph{not} use positivity of the ground state.

If in addition the ground state is strictly positive almost everywhere, then the
equality case can be characterized completely: saturation occurs if and only if
all one-body traps are harmonic with one common curvature, up to shifts of their
centers. Translation-invariant interactions remain arbitrary.

\medskip
The positivity assumption is therefore a real
structural boundary in the paper. It is only needed for the clean converse
direction in the equality theorem, not for the sharp bound itself.

We also record two quantitative statements that make
the rigidity aspect more explicit. The first is a purely spectral estimate: the
variance deficit controls the portion of $X\psi_0$ lying above the first active
energy shell. The second is structural: the defect
\[
G(x_1,\dots,x_N)
=\sum_{i=1}^N V_i'(x_i)-\Omega^2\sum_{i=1}^N m_i x_i,
\qquad
\Omega=\frac{\Delta_X}{\hbar},
\]
controls both the same spectral tail and the variance deficit itself. Thus near
saturation forces approximate concentration of $X\psi_0$ in the first active
shell and approximate harmonicity of the total center-of-mass forcing
profile in the $L^2(\rho_0)$ sense.

From the physical point of view, the theorem identifies a universal upper bound
on the static center-of-mass sensitivity of the ground state at fixed active
gap. In particular, if one probes the system by a weak rigid perturbation
proportional to $X$, then the relevant spectral input is precisely the
$X$-active sector rather than the full gap structure. We do not develop a full
linear-response formalism here, but the sharp bound shows that the ground-state
center-of-mass fluctuation cannot exceed the harmonic value at fixed active
gap, and that in the positivity class exact saturation is possible only for
common-frequency harmonic confinement.

The rest of the paper is organized as follows. Section 2 states the assumptions
and separates the general part of the theory from the positivity class.
Section 3 computes the universal center-of-mass commutators. Section 4 proves
the sharp variance bound. Section~5 treats the equality case and shows that
common-frequency harmonic traps are necessary and sufficient for saturation in
the positivity class. Section 6 proves sufficiency by exact
center-of-mass separation and completes the equality characterization.
Section 7 develops two quantitative rigidity estimates: one in terms of the
variance deficit and one in terms of the structural defect $G$.

\section{Assumptions, scope, and notation}

We consider
\begin{equation}\label{eq:H}
H=\sum_{i=1}^N \frac{p_i^2}{2m_i}+U\bigl(\{x_i-x_j\}_{i<j}\bigr)+\sum_{i=1}^N V_i(x_i),
\qquad
p_i=-i\hbar \partial_{x_i},
\end{equation}
acting on $L^2(\R^N)$.
The standing assumptions are the following.

\begin{enumerate}
\item[(A1)] $m_i>0$ for every $i$, the operator $H$ is self-adjoint and bounded
from below, and $H$ admits a normalized ground state $\psi_0$ with energy $E_0$.
Moreover, we work in the confining regime where $H$ has compact resolvent.
Equivalently, the spectrum of $H$ is purely discrete, each eigenvalue has finite
multiplicity, and there exists a complete orthonormal eigenbasis
$\{\psi_n\}_{n\ge 0}$ with
\[
H\psi_n=E_n\psi_n,
\qquad
E_0<E_1\le E_2\le \cdots,
\qquad
E_n\to +\infty.
\]

\item[(A2)] The interaction is translation invariant in the sense that it depends
only on relative coordinates:
\[
U(x_1,\dots,x_N)=U\bigl(\{x_i-x_j\}_{i<j}\bigr).
\]

\item[(A3)] The external traps satisfy $V_i\in C^1(\R)$. In addition, there
exists a dense subspace $\cD\subset L^2(\R^N)$ such that:
\begin{itemize}
\item $\cD$ is a common invariant core for $H$, for every $x_i$ and $p_i$, and
for the operators $X$ and $P$;
\item $\psi_0\in \cD$ and $X\psi_0\in \operatorname{Dom}(H)$;
\item on $\cD$, the formal commutator identities used below hold as operator
identities:
\[
[H,X]=-\frac{i\hbar}{M}P,
\qquad
[X,P]=i\hbar,
\qquad
[H,P]=i\hbar \sum_{i=1}^N V_i'(x_i).
\]
\end{itemize}
\end{enumerate}

A concrete sufficient class is obtained by taking
$V_i\in C^2(\R)$ and a real-valued translation-invariant interaction
\[
U=U\bigl(\{x_i-x_j\}_{i<j}\bigr),
\]
all with derivatives of at most polynomial growth. Assume that the total potential
\[
W(x_1,\dots,x_N)
:=
U\bigl(\{x_i-x_j\}_{i<j}\bigr)
+
\sum_{i=1}^N V_i(x_i)
\]
is bounded from below and coercive. That is,
$W(x)\to+\infty$ as $|x|\to\infty$.
Then $H$ is essentially self-adjoint on
$\mathcal S(\R^N)$, bounded from below, and has compact resolvent.
Moreover, $\mathcal S(\R^N)$ is a common invariant core for the
operators entering the commutator identities. We do not attempt to optimize
the weakest possible assumptions, since our goal is to isolate the
center-of-mass rigidity mechanism rather than develop the most general
case.

For the equality theory we isolate one additional assumption.

\begin{enumerate}
\item[(P)] The ground state lies in the positivity class:
\[
\psi_0(x_1,\dots,x_N)>0
\qquad\text{for almost every }(x_1,\dots,x_N)\in \R^N.
\]
\end{enumerate}

Assumption (P) is not used in the sharp bound. It enters only when we
pass from a vector identity to a pointwise functional identity in the equality
analysis. This applies, for example, to the standard bosonic or distinguishable
positive-ground-state setting. For fermionic problems the bound still survives,
but the clean converse direction may require additional input because of nodal
sets. For general background on positivity-improving Schr\"odinger semigroups
and the operator framework for confining Schr\"odinger Hamiltonians, see
\cite{Hess1977Domination,Simon1979Kato,Simon1982SchrodingerSemigroups,ReedSimonIV}.
Direct nondegeneracy and invariant-cone criteria for Schr\"odinger
ground states are given in \cite{Faris1972InvariantCones,Goelden1977}.

\begin{remark}
The assumptions are formulated only at a level sufficient for the
commutator argument and the discrete active-gap decomposition. More singular
or non-compact settings would require a separate domain discussion and are not
considered here.
\end{remark}

\begin{remark}
The compact-resolvent assumption in \textup{(A1)} is made for presentation
clarity: it lets us write the spectral identities as sums over eigenstates,
which keeps the active-gap mechanism transparent. A more general formulation in
terms of the spectral measure of $H$ is possible, but we do not pursue that
level of abstraction here.
\end{remark}

Set
\[
M:=\sum_{i=1}^N m_i,
\qquad
X:=\frac1M\sum_{i=1}^N m_i x_i,
\qquad
P:=\sum_{i=1}^N p_i.
\]
By shifting the physical origin, we may and do assume
\[
\avg{X}=0.
\]
Then
\[
\Var_0(X)=\avg{X^2}.
\]

Let $\{\psi_n\}_{n\ge 0}$ be an orthonormal eigenbasis of $H$,
$H\psi_n=E_n\psi_n$, and define
\[
X_{n0}:=\ip{\psi_n}{X\psi_0}.
\]
The active gap of the observable $X$ is
\begin{equation}\label{eq:active-gap}
\Delta_X:=\min\bigl\{E_n-E_0:\ n\ge 1,\ X_{n0}\ne 0\bigr\}.
\end{equation}
Equivalently, $\Delta_X$ is the bottom of the spectrum of $H-E_0$ restricted to
the cyclic subspace generated by $X\psi_0$.
When $X\psi_0\neq 0$, this number is strictly positive in the confining regime.

\begin{remark}\label{rem:active-gap-interacting}
For example, consider two particles in common-frequency harmonic traps with a
nontrivial translation-invariant interaction $U(r)$, where
$r=x_1-x_2$ and $\mu=m_1m_2/M$. The Hamiltonian separates as
\[
H=\frac{P^2}{2M}+\frac12M\Omega^2X^2+H_{\mathrm{rel}},
\qquad
H_{\mathrm{rel}}
=\frac{p_r^2}{2\mu}+\frac12\mu\Omega^2r^2+U(r).
\]
Suppose that the first excitation energy $\delta$ of $H_{\mathrm{rel}}$
satisfies $\delta<\hbar\Omega$. Then the ordinary spectral gap is
$E_1-E_0=\delta$, whereas $X$ acts only in the center of mass sector and hence
\[
\Delta_X=\hbar\Omega.
\]
Thus a rigid center-of-mass probe does not see the lower relative excitation.
For instance, $U$ may be chosen so that the effective relative potential
$\frac12\mu\Omega^2r^2+U(r)$ is a confining symmetric double well with a
first tunnelling splitting below $\hbar\Omega$.
\end{remark}
\section{Universal center-of-mass commutators}

\begin{proposition}[Center-of-mass TRK identities]\label{prop:commutators}
For the Hamiltonian \eqref{eq:H},
\begin{equation}\label{eq:first-comm}
[H,X]=-\frac{i\hbar}{M}P,
\end{equation}
and therefore
\begin{equation}\label{eq:double-comm}
[X,[H,X]]=\frac{\hbar^2}{M}.
\end{equation}
Consequently,
\begin{equation}\label{eq:trx}
\sum_{n\ge 1}(E_n-E_0)|X_{n0}|^2=\frac{\hbar^2}{2M}.
\end{equation}
\end{proposition}

\begin{proof}
Since $X$ is multiplication by
\[
X=\frac1M\sum_{j=1}^N m_j x_j,
\]
we have $[U,X]=0$ and $[V_i(x_i),X]=0$.
Hence only the kinetic part contributes:
\[
[H,X]=\sum_{i=1}^N \left[\frac{p_i^2}{2m_i},X\right].
\]
Now
\[
\left[\frac{p_i^2}{2m_i},X\right]
=\frac1M\sum_{j=1}^N m_j\left[\frac{p_i^2}{2m_i},x_j\right].
\]
If $j\ne i$, the commutator vanishes. For $j=i$,
\[
\left[\frac{p_i^2}{2m_i},x_i\right]=-\frac{i\hbar}{m_i}p_i.
\]
Therefore
\[
\left[\frac{p_i^2}{2m_i},X\right]
=\frac{m_i}{M}\left(-\frac{i\hbar}{m_i}p_i\right)
=-\frac{i\hbar}{M}p_i,
\]
and summing over $i$ gives \eqref{eq:first-comm}.

The second identity is immediate:
\[
[X,[H,X]]
=-\frac{i\hbar}{M}[X,P].
\]
Since
\[
[X,P]
=\left[\frac1M\sum_{i=1}^N m_i x_i,\sum_{j=1}^N p_j\right]
=\frac1M\sum_{i,j}m_i[x_i,p_j]
=\frac1M\sum_{i=1}^N m_i\,i\hbar
=i\hbar,
\]
we obtain \eqref{eq:double-comm}.

Finally, the standard spectral identity gives
\[
\frac12\avg{[X,[H,X]]}
=\sum_{n\ge 1}(E_n-E_0)|X_{n0}|^2.
\]
Using \eqref{eq:double-comm} yields \eqref{eq:trx}.
\end{proof}

\section{General sharp center-of-mass theorem}

The bound below depends only on the universal center-of-mass commutator
algebra and the active-gap definition; no positivity assumption is used.

\begin{theorem}[General sharp center-of-mass bound]\label{thm:bound}
Under the assumptions above,
\begin{equation}\label{eq:sharp-bound}
\Var_0(X)\le \frac{\hbar^2}{2M\Delta_X}.
\end{equation}
Moreover, equality holds if and only if
\begin{equation}\label{eq:equality-subspace}
X\psi_0\in \Ran P_{E_0+\Delta_X},
\end{equation}
equivalently, if and only if all active spectral weight of $X\psi_0$ sits at the
single energy $E_0+\Delta_X$.
\end{theorem}

\begin{proof}
Because $\avg{X}=0$,
\[
\Var_0(X)=\avg{X^2}=\sum_{n\ge 1}|X_{n0}|^2.
\]
By the definition of $\Delta_X$,
\[
E_n-E_0\ge \Delta_X
\qquad
\text{whenever }
X_{n0}\ne 0.
\]
Hence
\[
\Var_0(X)
=\sum_{n\ge 1}|X_{n0}|^2
\le \frac1{\Delta_X}\sum_{n\ge 1}(E_n-E_0)|X_{n0}|^2.
\]
Using \eqref{eq:trx} proves \eqref{eq:sharp-bound}.

Equality occurs precisely when every term with $X_{n0}\ne 0$ satisfies
$E_n-E_0=\Delta_X$, which is exactly \eqref{eq:equality-subspace}.
\end{proof}

\begin{remark}
Since $\Delta_X\ge E_1-E_0$, the weaker but more elementary estimate
\[
\Var_0(X)\le \frac{\hbar^2}{2M(E_1-E_0)}
\]
also follows. The point of \eqref{eq:sharp-bound} is that the active gap is the
correct one for the observable $X$.
\end{remark}

\section{Positivity-class equality theory}

We next turn to the converse direction. From this point on, the additional
positivity assumption \textup{(P)} is essential. It is what converts the vector
equality condition into a pointwise identity for the external forcing profile.

The next commutator is the many-body analogue of the one-particle identity
$[H,[H,x]]=(\hbar^2/m)V'(x)$.

\begin{proposition}[Second double commutator]\label{prop:second-double}
One has
\begin{equation}\label{eq:HHX}
[H,[H,X]]=\frac{\hbar^2}{M}\sum_{i=1}^N V_i'(x_i).
\end{equation}
\end{proposition}

\begin{proof}
From \eqref{eq:first-comm},
\[
[H,[H,X]]=-\frac{i\hbar}{M}[H,P].
\]
Now
\[
[H,P]
=\left[\sum_{i=1}^N \frac{p_i^2}{2m_i},P\right]+[U,P]+\sum_{i=1}^N [V_i(x_i),P].
\]
The kinetic commutator vanishes because all momenta commute. Also,
\[
[U,P]=i\hbar \sum_{k=1}^N \partial_{x_k}U=0,
\]
since $U$ depends only on differences $x_i-x_j$.
Finally,
\[
[V_i(x_i),P]=[V_i(x_i),p_i]=i\hbar V_i'(x_i).
\]
Therefore
\[
[H,P]=i\hbar\sum_{i=1}^N V_i'(x_i),
\]
and \eqref{eq:HHX} follows.
\end{proof}

\begin{theorem}[Positivity-class necessity of common-frequency harmonic traps]\label{thm:necessity}
Assume \textup{(P)}.
If equality holds in \eqref{eq:sharp-bound}, then with
\[
\Omega:=\frac{\Delta_X}{\hbar}
\]
there exist constants $a_i,C_i\in \R$ such that
\begin{equation}\label{eq:harmonic-form}
V_i(x)=\frac12 m_i\Omega^2(x-a_i)^2+C_i,
\qquad i=1,\dots,N,
\end{equation}
and
\begin{equation}\label{eq:centering-condition}
\sum_{i=1}^N m_i a_i=0.
\end{equation}
Equivalently,
\[
V_i'(x)=m_i\Omega^2 x+c_i,
\qquad
\sum_{i=1}^N c_i=0.
\]
\end{theorem}

The positivity assumption is part of the equality theorem, not of the
bound. Thus the estimate remains valid for fermionic or nodal states whenever
the spectral assumptions hold, but the converse implication from saturation to
pointwise harmonicity is not asserted in that generality, because the argument
requires division by $\psi_0$ almost everywhere.

\begin{proof}
If equality holds in Theorem \ref{thm:bound}, then
$X\psi_0\in \Ran P_{E_0+\Delta_X}$.
Since $\avg{X}=0$, the vector $X\psi_0$ has no ground-state component.
Therefore, if equality holds and all spectral weight of $X\psi_0$ lies in the
shell $E_0+\Delta_X$, then
\[
(H-E_0)^2X\psi_0=\Delta_X^2 X\psi_0.
\]
But
\[
(H-E_0)X\psi_0=[H,X]\psi_0,
\]
hence
\[
(H-E_0)^2X\psi_0=[H,[H,X]]\psi_0.
\]
Using \eqref{eq:HHX}, we obtain
\[
\frac{\hbar^2}{M}\left(\sum_{i=1}^N V_i'(x_i)\right)\psi_0
=\Delta_X^2 X\psi_0
=\frac{\Delta_X^2}{M}\left(\sum_{i=1}^N m_i x_i\right)\psi_0.
\]

Since $\psi_0>0$ almost everywhere by assumption \textup{(P)}, the zero
set of $\psi_0$ has measure zero, and the coefficient multiplying $\psi_0$ must
vanish almost everywhere. Hence

\[
\sum_{i=1}^N \bigl(V_i'(x_i)-m_i\Omega^2 x_i\bigr)=0
\qquad\text{a.e. on }\R^N.
\]
Set
\[
g_i(x):=V_i'(x)-m_i\Omega^2 x.
\]
Then
\[
\sum_{i=1}^N g_i(x_i)=0
\qquad
\text{for a.e. }(x_1,\dots,x_N)\in\R^N.
\]
Fix $i$.
By Fubini's theorem, for almost every choice of the remaining variables
\[
(x_1,\dots,x_{i-1},x_{i+1},\dots,x_N),
\]
the above identity holds for almost every $x_i$.
Hence, for such a choice,
\[
g_i(x_i)=-\sum_{j\neq i} g_j(x_j)
\]
for almost every $x_i$, and the right-hand side is independent of $x_i$.
Therefore $g_i$ is almost everywhere equal to a constant.
Since $V_i\in C^1(\R)$, the function $g_i$ is continuous, so it is in fact
constant everywhere on $\R$.
Thus
\[
V_i'(x)=m_i\Omega^2 x+c_i
\]
for some constants $c_i$, and summing over $i$ shows $\sum_i c_i=0$.
Integrating gives
\[
V_i(x)=\frac12 m_i\Omega^2 x^2+c_i x+C_i^{\ast}
=\frac12 m_i\Omega^2 (x-a_i)^2+C_i^{\ast},
\qquad
a_i:=-\frac{c_i}{m_i\Omega^2}.
\]
The condition $\sum_i c_i=0$ becomes $\sum_i m_i a_i=0$, which is
\eqref{eq:centering-condition}.
\end{proof}

\begin{remark}
This is the only point where positivity is used.
\end{remark}

\section{Sufficiency and exact separation}

\begin{theorem}[Sufficiency via exact center-of-mass separation]\label{thm:sufficiency}
Suppose that for some $\Omega>0$ and some constants $a_i,C_i$,
\[
V_i(x)=\frac12 m_i\Omega^2(x-a_i)^2+C_i,
\qquad
\sum_{i=1}^N m_i a_i=0.
\]
Then the Hamiltonian separates as
\[
H=H_{\mathrm{COM}}+H_{\mathrm{rel}},
\]
where
\[
H_{\mathrm{COM}}=\frac{P^2}{2M}+\frac12 M\Omega^2 X^2.
\]
In particular,
\[
\Delta_X=\hbar\Omega,
\qquad
\Var_0(X)=\frac{\hbar}{2M\Omega}=\frac{\hbar^2}{2M\Delta_X},
\]
so the bound \eqref{eq:sharp-bound} is saturated.
\end{theorem}

\begin{proof}
Write
\[
x_i=X+y_i,
\qquad
\sum_{i=1}^N m_i y_i=0.
\]
Then
\[
\sum_{i=1}^N m_i(x_i-a_i)^2
=\sum_{i=1}^N m_i(X+y_i-a_i)^2
=MX^2+\sum_{i=1}^N m_i(y_i-a_i)^2,
\]
because both $\sum_i m_i y_i$ and $\sum_i m_i a_i$ vanish.
In addition, the kinetic energy admits the standard center-of-mass/relative
decomposition, for instance in Jacobi coordinates,
\[
\sum_{i=1}^N \frac{p_i^2}{2m_i}
=\frac{P^2}{2M}+T_{\mathrm{rel}},
\]
where $T_{\mathrm{rel}}$ depends only on the relative variables and their
conjugate momenta.
Therefore
\[
\sum_{i=1}^N V_i(x_i)
=\frac12 M\Omega^2 X^2+W_{\mathrm{rel}}(y_1,\dots,y_N)+C.
\]
Since $U$ depends only on the differences $x_i-x_j=y_i-y_j$, the full
Hamiltonian splits into
\[
H_{\mathrm{COM}}=\frac{P^2}{2M}+\frac12 M\Omega^2 X^2
\]
plus a relative Hamiltonian $H_{\mathrm{rel}}$.

Hence the ground state factorizes,
\[
\psi_0(X,\mathrm{rel})=\phi_0(X)\chi_0(\mathrm{rel}),
\]
where $\phi_0$ is the ground state of the one-dimensional harmonic oscillator
$H_{\mathrm{COM}}$.
Since the operator $X$ acts only on the center-of-mass coordinate, it acts
trivially on the relative factor, and therefore
\[
X\psi_0=(X\phi_0)\chi_0.
\]
For a one-dimensional harmonic oscillator,
\[
X\phi_0\in \Ran P^{\mathrm{COM}}_{\hbar\Omega},
\]
indeed $X\phi_0$ is proportional to the first excited center-of-mass mode.
Therefore the active gap for $X$ is exactly $\Delta_X=\hbar\Omega$, and
\[
\Var_0(X)=\frac{\hbar}{2M\Omega}
=\frac{\hbar^2}{2M\Delta_X}.
\]
\end{proof}

\begin{remark}
The sufficiency mechanism is exactly the center-of-mass decoupling familiar from
Kohn-type theorems. Once all trap curvatures coincide, interactions disappear
from the center-of-mass sector and the rigid mode becomes a pure harmonic
oscillator with frequency $\Omega$.
\end{remark}

Combining Theorems \ref{thm:necessity} and \ref{thm:sufficiency}, we obtain the
clean many-body equality statement.

\begin{theorem}[Positivity-class converse and equality characterization]\label{thm:iff}
Under the standing assumptions, the sharp bound
\[
\Var_0(X)\le \frac{\hbar^2}{2M\Delta_X}
\]
always holds. If, in addition, the ground state is strictly positive almost
everywhere, that is, if \textup{(P)} holds, then equality occurs if and only if,
after a global translation of the origin, the one-body traps are harmonic with
a common frequency,
\[
V_i(x)=\frac12 m_i\Omega^2(x-a_i)^2+C_i,
\qquad
\sum_{i=1}^N m_i a_i=0,
\qquad
\Omega=\frac{\Delta_X}{\hbar}.
\]
The interaction $U$ is otherwise arbitrary, subject only to translation
invariance.
\end{theorem}

\section{Quantitative near-saturation theory}

We finish with two elementary quantitative consequences of the same
spectral decomposition. The first one relates the variance deficit to the
spectral weight of $X\psi_0$ outside the first active shell. The second one
connects this tail to the ground-state density-weighted forcing defect $G$.

\subsection*{Variance deficit and concentration in the active shell}

Let
\[
Q_X:=P_{E_0+\Delta_X}
\]
be the spectral projection onto the first active shell. Define the variance
deficit
\begin{equation}\label{eq:deficit}
\cD_X:=\frac{\hbar^2}{2M\Delta_X}-\Var_0(X)\ge 0.
\end{equation}
To measure the separation between the first active shell and the rest of the
active spectrum, set
\begin{equation}\label{eq:GammaX}
\Gamma_X:=
\inf\bigl\{(E_n-E_0)-\Delta_X:\ n\ge 1,\ E_n-E_0>\Delta_X,\ X_{n0}\ne 0\bigr\},
\end{equation}
with the convention $\Gamma_X=+\infty$ if no such index exists. In that case
all higher active-shell tails vanish identically, and the bounds below are read
with $1/\Gamma_X=0$.

\begin{proposition}[Deficit identity and active-shell concentration]
\label{prop:deficit}
One has the exact identity
\begin{equation}\label{eq:deficit-identity}
\cD_X
=\frac1{\Delta_X}\sum_{n\ge 1}\bigl((E_n-E_0)-\Delta_X\bigr)|X_{n0}|^2.
\end{equation}
Consequently,
\begin{equation}\label{eq:spectral-tail-deficit}
\|(I-Q_X)X\psi_0\|^2
=\sum_{\substack{n\ge 1\\ E_n-E_0>\Delta_X}} |X_{n0}|^2
\le \frac{\Delta_X}{\Gamma_X}\,\cD_X.
\end{equation}
If the first active shell is one-dimensional, say
$Q_X=\ip{\phi_X}{\cdot}\phi_X$ with $\|\phi_X\|=1$, then
\begin{equation}\label{eq:one-dimensional-active-shell}
1-\left|\left\langle \phi_X,\frac{X\psi_0}{\|X\psi_0\|}\right\rangle\right|^2
\le
\frac{\Delta_X}{\Gamma_X}\,\frac{\cD_X}{\Var_0(X)}.
\end{equation}
\end{proposition}

\begin{proof}
If $\Gamma_X=+\infty$, then $X\psi_0$ is already entirely supported in the first
active shell and there is nothing to prove. We therefore assume below that
$\Gamma_X<\infty$.

By Proposition \ref{prop:commutators},
\[
\frac{\hbar^2}{2M}
=\sum_{n\ge 1}(E_n-E_0)|X_{n0}|^2.
\]
Subtracting $\Delta_X\Var_0(X)=\Delta_X\sum_{n\ge 1}|X_{n0}|^2$ and dividing by
$\Delta_X$ yields \eqref{eq:deficit-identity}.

Now
\[
\|(I-Q_X)X\psi_0\|^2
=\sum_{\substack{n\ge 1\\ E_n-E_0>\Delta_X}} |X_{n0}|^2.
\]
For every term in this sum,
\[
(E_n-E_0)-\Delta_X\ge \Gamma_X.
\]
Therefore
\[
\Delta_X\cD_X
=\sum_{\substack{n\ge 1\\ E_n-E_0>\Delta_X}}
\bigl((E_n-E_0)-\Delta_X\bigr)|X_{n0}|^2
\ge
\Gamma_X\|(I-Q_X)X\psi_0\|^2,
\]
which proves \eqref{eq:spectral-tail-deficit}.

If $Q_X$ is one-dimensional, then
\[
\left\|(I-Q_X)\frac{X\psi_0}{\|X\psi_0\|}\right\|^2
=1-\left|\left\langle \phi_X,\frac{X\psi_0}{\|X\psi_0\|}\right\rangle\right|^2.
\]
Using $\|X\psi_0\|^2=\Var_0(X)$ and \eqref{eq:spectral-tail-deficit} gives
\eqref{eq:one-dimensional-active-shell}.
\end{proof}

\begin{remark}
Proposition~\ref{prop:deficit} is purely spectral. It shows that near-saturation
of the sharp bound forces $X\psi_0$ to lie almost entirely in the first active
shell, even before any structural information about the potentials is used.
\end{remark}

\subsection*{Structural defect and quantitative rigidity}

The natural structural defect is
\begin{equation}\label{eq:G}
G(x_1,\dots,x_N):=\sum_{i=1}^N V_i'(x_i)-\Omega^2\sum_{i=1}^N m_i x_i,
\qquad
\Omega=\frac{\Delta_X}{\hbar}.
\end{equation}
This is the many-body analogue of the one-particle rigidity object
$V'(x)-m\Omega^2 x$.

Accordingly, $\|G\|_{L^2(\rho_0)}$ measures approximate harmonicity only
of the total center-of-mass forcing profile in the ground-state
density-weighted sense. This is not a uniform control in configuration space.
Regions where $\rho_0$ is small are weakly tested by this defect.

\begin{proposition}[Structural defect controls deficit and spectral tail]
\label{prop:tail}
Let
\[
\Phi_X:=\bigl((H-E_0)^2-\Delta_X^2\bigr)X\psi_0.
\]
Then
\begin{equation}\label{eq:phi-G}
\Phi_X=\frac{\hbar^2}{M}G\psi_0
\end{equation}
and
\begin{equation}\label{eq:phi-norm}
\|\Phi_X\|^2=\frac{\hbar^4}{M^2}\|G\|_{L^2(\rho_0)}^2,
\qquad
\rho_0:=|\psi_0|^2.
\end{equation}
Moreover,
\begin{equation}\label{eq:deficit-bound-G}
\cD_X
\le
\frac{\hbar^4}{M^2\Delta_X\Gamma_X(2\Delta_X+\Gamma_X)^2}\,
\|G\|_{L^2(\rho_0)}^2,
\end{equation}
and
\begin{equation}\label{eq:tail-bound}
\|(I-Q_X)X\psi_0\|^2
\le
\frac{\hbar^4}{M^2\Gamma_X^2(2\Delta_X+\Gamma_X)^2}\,
\|G\|_{L^2(\rho_0)}^2.
\end{equation}
\end{proposition}

\begin{proof}
If $\Gamma_X=+\infty$, then the active spectral tail above $\Delta_X$ vanishes
identically and the conclusions are immediate from \eqref{eq:phi-norm}. We may
therefore assume $\Gamma_X<\infty$.

By the same computation as before,
\[
(H-E_0)^2X\psi_0=[H,[H,X]]\psi_0.
\]
Using \eqref{eq:HHX} and the definition of $\Omega$,
\[
\Phi_X
=\left(\frac{\hbar^2}{M}\sum_{i=1}^N V_i'(x_i)-\Delta_X^2 X\right)\psi_0
=\frac{\hbar^2}{M}\left(\sum_{i=1}^N V_i'(x_i)-\Omega^2\sum_{i=1}^N m_i x_i\right)\psi_0,
\]
which is \eqref{eq:phi-G}. Squaring gives \eqref{eq:phi-norm}.

On the spectral side,
\[
\|\Phi_X\|^2
=\sum_{n\ge 1}\bigl((E_n-E_0)^2-\Delta_X^2\bigr)^2|X_{n0}|^2.
\]
Write
\[
\mu_n:=(E_n-E_0)-\Delta_X\ge 0.
\]
Then
\[
(E_n-E_0)^2-\Delta_X^2=\mu_n(2\Delta_X+\mu_n).
\]
For terms with $\mu_n>0$ we have $\mu_n\ge \Gamma_X$, hence
\[
\mu_n^2(2\Delta_X+\mu_n)^2
\ge \Gamma_X(2\Delta_X+\Gamma_X)^2\,\mu_n
\]
and also
\[
\mu_n^2(2\Delta_X+\mu_n)^2
\ge \Gamma_X^2(2\Delta_X+\Gamma_X)^2.
\]
Therefore,
\[
\|\Phi_X\|^2
\ge
\Gamma_X(2\Delta_X+\Gamma_X)^2
\sum_{n\ge 1}\mu_n |X_{n0}|^2
=
\Gamma_X\Delta_X(2\Delta_X+\Gamma_X)^2\,\cD_X,
\]
where we used \eqref{eq:deficit-identity} in the last step. This proves
\eqref{eq:deficit-bound-G}.

Similarly,
\[
\|\Phi_X\|^2
\ge \Gamma_X^2(2\Delta_X+\Gamma_X)^2
\sum_{\substack{n\ge 1\\ E_n-E_0>\Delta_X}}|X_{n0}|^2
=
\Gamma_X^2(2\Delta_X+\Gamma_X)^2\|(I-Q_X)X\psi_0\|^2,
\]
which yields \eqref{eq:tail-bound}.
\end{proof}

\begin{remark}
The two estimates in Proposition \ref{prop:tail} complement each other.
Equation \eqref{eq:deficit-bound-G} says that small structural defect forces the
variance to be close to the sharp upper bound. Equation \eqref{eq:tail-bound}
says that it forces the non-extremal part of $X\psi_0$ to leave the higher
active shells. In the positivity class, exact vanishing of $G$ recovers the
common-frequency harmonic structure of Theorem \ref{thm:iff}.

The constants are useful only when the next active shell is separated
from the first one. If $\Gamma_X$ is small, the estimates correctly become
weak. Near-saturation can then involve redistribution of spectral weight
between almost degenerate active shells.
\end{remark}

\section{An illustrative exactly solvable example}

It is useful to record one benchmark. For comparison
with the interacting example in Remark~\ref{rem:active-gap-interacting}, take
two equal masses $m$ with
$U=0$, $V_1(x)=\frac12m\omega^2x^2$, and
$V_2(x)=\frac12m\lambda^2\omega^2x^2$. For
$X=(x_1+x_2)/2$ and $M=2m$, a direct Gaussian calculation gives
\[
\Var_0(X)=\frac{\hbar}{8m\omega}\left(1+\frac1\lambda\right),
\qquad
\Delta_X=\hbar\omega\min\{1,\lambda\}.
\]
Hence
\begin{equation}\label{eq:explicit-ratio}
\mathcal R(\lambda)
:=\frac{2M\Delta_X}{\hbar^2}\Var_0(X)
=\frac12\bigl(1+\min\{\lambda,\lambda^{-1}\}\bigr)\le 1,
\end{equation}
with equality only if $\lambda=1$. This compact benchmark detects the
matching-curvature condition, whereas Remark~\ref{rem:active-gap-interacting}
supplies the genuinely interacting example and shows that the active and
ordinary gaps may differ.

\section{Conclusion}

The center-of-mass observable in a one-dimensional many-body system with
translation-invariant interactions enjoys the exact algebra
\[
[H,X]=-\frac{i\hbar}{M}P,
\qquad
[X,[H,X]]=\frac{\hbar^2}{M},
\qquad
[H,[H,X]]=\frac{\hbar^2}{M}\sum_{i=1}^N V_i'(x_i).
\]
This yields the general sharp bound
\[
\Var_0(X)\le \frac{\hbar^2}{2M\Delta_X}.
\]
Under the positivity hypothesis \textup{(P)}, the converse direction also
closes. Saturation is equivalent to the fact that every one-body trap is
harmonic with the same frequency $\Omega=\Delta_X/\hbar$, up to shifts of the
individual trap centers, while the interaction remains arbitrary inside the
translation-invariant class.

The near-saturation theory is encoded by two complementary quantities. The
variance deficit
\[
\cD_X=\frac{\hbar^2}{2M\Delta_X}-\Var_0(X)
\]
controls the weight of $X\psi_0$ outside the first active shell. The structural
defect
\[
G=\sum_{i=1}^N V_i'(x_i)-\Omega^2\sum_{i=1}^N m_i x_i,
\]
controls both $\cD_X$ and the same spectral tail through the ground-state
density. These estimates should be viewed as quantitative consequences of the
exact commutator mechanism.

The conceptual payoff is therefore the following: for the rigid interacting
center-of-mass mode, harmonic confinement is not merely a convenient solvable
case and not merely a setting in which decoupling happens dynamically. In the
class considered here, it is the sharp static extremizer at fixed active gap,
and in the positivity class it is the unique extremizer.

\printbibliography

\end{document}